\documentclass[conference]{IEEEtran}
\usepackage[latin9]{inputenc}

\usepackage{graphicx}
\usepackage{amssymb}
\usepackage [noadjust]{cite} 
\usepackage{bm}  
\usepackage{amsmath}
\usepackage{amssymb}
\usepackage{enumerate}
\usepackage{stfloats}
\usepackage{cases}
\usepackage{epstopdf}
\usepackage{soul,color} 
\usepackage[percent]{overpic}
\usepackage{tabularx}
\usepackage{subfigure}
\usepackage[table]{xcolor}
\usepackage{multirow}
\usepackage{booktabs}  
\usepackage{tabu} 

\makeatletter
\newcommand{\dbloverline}[1]{\overline{\dbl@overline{#1}}}
\newcommand{\dbl@overline}[1]{\mathpalette\dbl@@overline{#1}}
\newcommand{\dbl@@overline}[2]{%
	\begingroup
	\sbox\z@{$\m@th#1\overline{#2}$}%
	\ht\z@=\dimexpr\ht\z@-2\dbl@adjust{#1}\relax
	\box\z@
	\ifx#1\scriptstyle\kern-\scriptspace\else
	\ifx#1\scriptscriptstyle\kern-\scriptspace\fi\fi
	\endgroup
}
\newcommand{\dbl@adjust}[1]{%
	\fontdimen8
	\ifx#1\displaystyle\textfont\else
	\ifx#1\textstyle\textfont\else
	\ifx#1\scriptstyle\scriptfont\else
	\scriptscriptfont\fi\fi\fi 3
}
\makeatother

 \newcommand{\Rdddirectone}[3]{R^{dir_1}}
 \newcommand{\Rdddirectzero}[2]{R^{dir_0}}

 \newcommand{\Rddrelay}[2]{R^{rel}}

 \newcommand{\Rcellularzero}[2]{R^{cel_0}}
 \newcommand{\Rcellularone}[3]{R^{cel_1}}

\newtheorem {theorem}{Theorem}

\newtheorem {lemma}{Lemma}

\begin{document}
\IEEEoverridecommandlockouts


\raggedbottom
		
\title{On the Performance of Non-Orthogonal Multiple Access (NOMA): Terrestrial vs. Aerial Networks}

\author{
	Mehdi Monemi$^*$,
		Hina Tabassum$^{**}$, and
		Ramein Zahedi$^{**}$
	\\
	$^*$Department of Electrical  Engineering,
		Salman Farsi University of Kazerun, Kazerun, Iran
	\\
	$^{**}$Department of Electrical Engineering and Computer Science, York University, Canada
	\\
		(monemi@kazerunsfu.ac.ir, hinat@yorku.ca, rgz96@my.yorku.ca)
	 \thanks{M. Monemi was a visiting research Professor in the York University, Canada.  This work was supported by the Discovery
		Grant from the Natural Sciences and Engineering Research Council of Canada
		(NSERC). \vspace{3mm}\newline978-1-7281-5320-9/20/\$31.00~
		\copyright~2020 IEEE}
}

\maketitle

	\begin{abstract}
	Non-orthogonal multiple access (NOMA) is  a promising multiple access technique for beyond fifth generation  (B5G) cellular wireless networks, where several users can be served on a single time-frequency resource block, using the concepts of superposition coding at the transmitter and self-interference cancellation (SIC) at the receiver. For terrestrial networks, the achievable performance gains  of  NOMA  over  traditional  orthogonal  multiple access  (OMA)  are  well-known. However,  the achievable performance of NOMA  in  aerial networks, compared to terrestrial networks, is not well-understood. In this paper, we  provide a unified analytic  framework to characterize the outage probabilities of users  considering various network settings, such as i) uplink and downlink NOMA and OMA in aerial networks, and  ii) uplink and downlink NOMA and OMA in terrestrial networks. In particular,  we derive  closed-form rate outage probability expressions for two users, considering line-of-sight (LOS) Rician fading channels.    Numerical results validate the derived analytical expressions and  demonstrate the difference of   outage probabilities  of users with OMA and NOMA transmissions. Numerical results  unveil that the optimal UAV height  increases with the increase in Rice-$K$ factor, which implies strong line-of-sight (LOS) conditions.

	\end{abstract}
\begin{keywords}
    Non-orthogonal multiple access (NOMA), Unmanned aerial vehicle (UAV), terrestrial/aerial networks, SINR, outage probability.
\end{keywords}
	

\thispagestyle{empty}

\section{Introduction}
Unmanned aerial vehicle (UAV)-enabled wireless communications is  indispensable for seamless functioning of the emerging fifth generation or sixth generation networks \cite{li2018uav, hossain2014evolution}.  In contrast to terrestrial cellular networks, the distinct features of UAV networks include wider coverage, three dimensional flexible deployment, line-of-sight (LOS) transmissions, and swift on-demand deployment and removal of UAVs \cite{sharma2017uav, sekander2018multi}. The performance gains of UAV-enabled wireless communications are quite evident in the existing literature with orthogonal multiple access (OMA) transmissions. Nevertheless, to achieve massive connectivity in aerial networks, non-orthogonal multiple access (NOMA) is another potential technique that can serve multiple ground users at the same time/frequency/code domain, but with different power levels \cite{originalNOMA, tabassum2017uplink}. 
However, the gains of NOMA in aerial networks (when compared to terrestrial networks) are not comprehensively explored taking into account the distinct channel features such as 1) line-of-sight (LoS) Rician fading in aerial networks compared to non-LOS (NLOS) Rayleigh fading in terrestrial networks, and 2) aerial and terrestrial path-loss models.


In what follows, we highlight the existing literature focusing on the uplink and downlink NOMA in UAV networks. In \cite{sharma2017uav}, the authors have considered a downlink UAV-assisted NOMA network  of two users. By adopting Rician fading for LOS UAV-to-ground links, the relative performance between NOMA and OMA was analyzed. In \cite{rupasinghe2017non}, the authors have focused on optimizing the altitude of the UAV by employing the outage sum rates in a multi-antenna UAV downlink NOMA network. In \cite{nasir2019uav}, for downlink UAV NOMA networks, the problem of max-min rate optimization under the constraints related to UAV altitude, the amount of power, and bandwidth allocated to users was considered. In \cite{sohail2018non}, a user pairing and power allocation scheme was presented to maximize the sum-rate of users and reduce the energy consumption of the UAV.   In \cite{cui2019multiple},  joint trajectory design and resource allocation problem was formulated and solved considering both OMA and NOMA modes. In \cite{sun2018cyclical}, the cyclical NOMA was introduced for UAV networks wherein the UAV's flight cycle was divided into several time slots and the minimum throughput of all ground users was maximized, by jointly optimizing user scheduling  and UAV trajectory.

Very few existing works have studied the performance of  uplink NOMA in UAV networks due to the distinct desired and interference channel statistics. In \cite{farajzadeh2019uav}, the authors studied uplink UAV-assisted backscatter networks, wherein the UAV acts both as a mobile power transmitter and as an information collector. A resource management scheme was proposed  to maximize the number of successfully decoded bits while minimizing the UAV's flight time and optimizing its altitude. In \cite{duan2019resource}, an uplink NOMA transmission scheme for multi-UAV aided IoT networking system was proposed, wherein the objective was to maximize the total uplink capacity while optimizing the subchannel assignment, the uplink transmit power of IoT nodes, and the altitude of UAVs. In particular, a clustering method was proposed to group IoT nodes and a low complexity algorithm was designed for efficient subchannel assignment based on the results from clustering. Finally, in \cite{UplinkNOMArandomaccess}, maximum stable throughput for uplink NOMA was investigated. By expressing the stability condition in terms of UAV's hovering altitude, beamwidth, and traffic intensity, a stabilizing algorithm was devised which recursively controls users' access, and adjusts the altitude and beamwidth.

As in the aforementioned research works, the gains of NOMA over OMA are typically investigated in either downlink or uplink. Different from the literature, this paper  provides a unified analytic  framework to characterize the outage probabilities of users   considering various network settings, such as 1) uplink and downlink NOMA and OMA in aerial networks, and  2) uplink and downlink NOMA and OMA in terrestrial networks. The rate outage probability expressions  consider line-of-sight (LOS) Rician fading for aerial transmissions  and non-line-of-sight (NLOS) Rayleigh fading for terrestrial transmissions.     Numerical results validate the derived analytical expressions and  demonstrate the difference of outage probabilities  of users in both the aerial  and  terrestrial networks. 
Numerical results  unveil that the optimal UAV height  increases with the increase in Rice-$K$ factor, which implies strong line-of-sight (LOS) conditions. Also, the results demonstrate that  the NOMA gains are evident for low to moderate values of target spectral efficiency in aerial networks. On the other hand, in terrestrial networks, NOMA gains can be observed for higher values of target rate threshold.

\section{System Model and Assumptions}
\subsection{System Model}
We consider a NOMA network  for the downlink and uplink transmissions. For either of the uplink and downlink transmission, we consider either terrestrial or aerial UAV base stations (BSs). Fig.~\ref{fig:uplink_downlink}  shows the downlink and uplink NOMA for UAV network, respectively. The terrestrial or UAV BS communicates with two  ground users $U_1$ and $U_2$ according to the NOMA transmission principle.  We assume that the spectrum resources at UAV for the uplink communication with the users and for the backhaul are orthogonal. Due to high complexity of successive interference cancellation (SIC) with multiple users in NOMA, we focus on the two-user case  and defer multi-user case for future works. 

We consider three-dimensional Cartesian coordinates $(x,y,z)$ where the ground plane is represented by $(x,y,0)$.
We have $N$ users who are randomly located in the coverage area of the UAV (or terrestrial BS). Out of $N$ users, two users $U_1$ and $U_2$ are the users with minimum and maximum distance $r_{min}$ and $r_{max}$ from the cell-center, respectively, and thus,  their distances from the UAV can be given as follows:
\begin{subequations}\label{eq:EUCDIST}
\begin{align}
    d_1=d_{\mathrm{min}} &= (h^2 + r_{\mathrm{min}}^2)^{\frac{1}{2}},\label{eq:first_user}\\
    d_2=d_{\mathrm{max}} &= (h^2 + r_{\mathrm{max}}^2)^{\frac{1}{2}}.\label{eq:second_user}
\end{align}
\end{subequations}

\begin{figure}
    \centering
    \includegraphics[width=\linewidth]{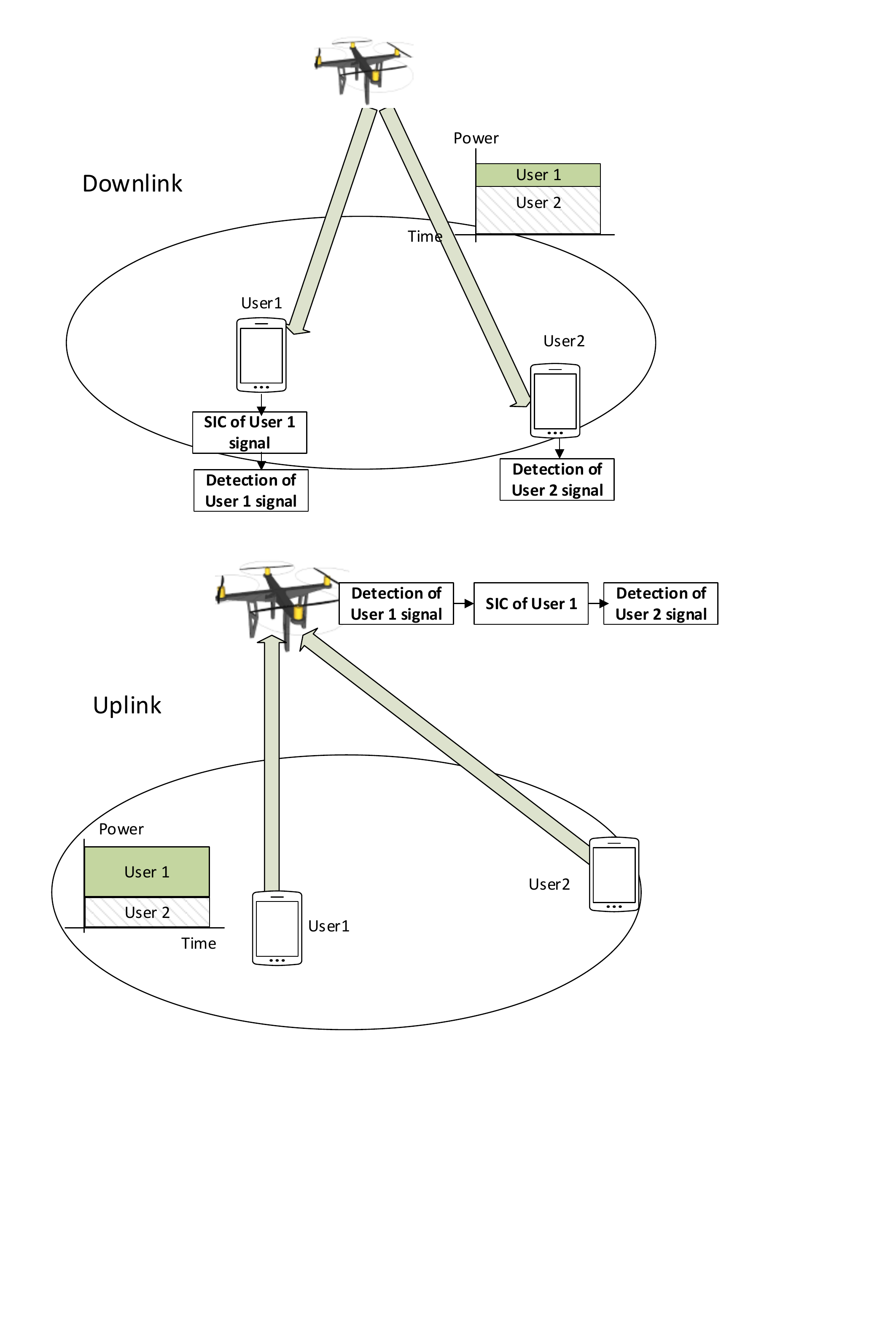}  \vspace{-30mm}
    \caption{Graphical illustration of uplink and downlink NOMA in aerial networks.}
    \label{fig:uplink_downlink}
\end{figure}
Based on the above, we can generalize the Euclidean distance from the ground users $U_1$ and $U_2$ to UAV, respectively, as follows:
\begin{align}
    d_i &= (h^2 + r_i^2)^{\frac{1}{2}}, \ i=1,2
\end{align}
where $r_i=\sqrt{x_i^2 + y_i^2}$. For the case of terrestrial BS, we can simply substitute $h=0$ and thus we have $d_i=r_i$.


\subsection{Channel Model}
Following \cite{wirelessrelaynetwork}, the wireless channels between ground users and UAV are assumed to experience large-scale path-loss and small-scale quasi-static frequency non-selective fading. 
\subsubsection{Path-Loss and Fading for Terrestrial Links} For the case of terrestrial BS, the fading for user $i$ is generally assumed to be Rayleigh  denoted by $\Phi_i$ whose probability density function (pdf) and cumulative distribution function (cdf) of (power) channel gain are expressed, respectively, as follows:
\begin{align}
    \label{eq:exponential_pdf}
    f_{\Phi_i}(x)
        &=\lambda_i e^{-\lambda_i x},
    \\
    \label{eq:exponential_cdf}
    F_{\Phi_i}(x)
        &=1-e^{-\lambda_i x},
\end{align}
where $\lambda_i$ is the fading parameter for user $i$.
The large-scale path-loss is expressed (in dB) as follows:
\begin{equation}
	L^t_{dB}(d_i) = 10\alpha\log_{10}\left(d_i\right) + \beta,
\end{equation}
where $\beta = 20\log_{10}\left(\frac{4 \pi f_c}{c}\right) + \eta$ in which $f_c$ is the carrier frequency, $c$ is the speed of light, and $\eta$ is a constant depending on the environment. The corresponding terrestrial channel gain is obtained by
\begin{align}
    \label{eq:terrestrial_gain}
    G^t(d_i)=10^{-0.1 L^t_{dB}(d_i)}.
\end{align}
\subsubsection{Path-Loss and Fading for Aerial Links} For the case of UAV BS, each user $i$ establishes aerial link with Rician fading  denoted by $\chi_i$ with two degrees of freedom.  The pdf of user $\chi_i$ is thus expressed as follows:
\begin{equation}
    f_{\chi_i}(x) = \frac{(1+K_i)\mathrm{e}^{-K_i}}{\Omega_i}\mathrm{e}^{-\frac{(1+K_i)x}{\Omega_i}}\mathcal{I}_o\left(2\sqrt{\frac{K_i(1+K_i)x}{\Omega_i}}\right)\mathrm{,}
\end{equation}
where $\mathcal{I}_o(\cdot)$ is the zeroth-order modified Bessel function of the first kind, and $\Omega_i$ and $K_i$ are referred, respectively, as the fading power and Rician shape factor for user $i$. The corresponding cdf is expressed as follows:
\begin{equation}
    \label{eq:rician_cdf}
    F_{\chi_i}(x) = 1 - \mathcal{Q}\left(\sqrt{2K_i}, \sqrt{\frac{2(1+K_i)x}{\Omega_i}}\right)\mathrm{,}
\end{equation}
where $\mathcal{Q}(a,b) \triangleq \int_{b}^{\infty} x\mathrm{e}^{-\frac{a^2+x^2}{2}}\mathcal{I}_o(ax) dx$ denotes the Marcum Q-function of first order \cite{sharma2017uav}. The large-scale aerial path-loss is expressed (in dB) as follows:
\begin{equation}
	L_{dB}^a(h,r_i) = 20\log\left(\sqrt{h^2+r_i^2}\right) + AP_{LOS}(h,r_i) + B.
\end{equation}
We have $A = \eta_{\mathrm{LOS}} - \eta_{\mathrm{NLOS}}$ and  $B = 20\log\left(\frac{4 \pi f_c}{c}\right) + \eta_{\mathrm{NLOS}}$ in which $\eta_{\mathrm{LOS}}\ \mathrm{and}\ \eta_{\mathrm{NLOS}}$ are the losses corresponding to the LOS and non-LOS reception depending on the environment.
The probability of LOS  is represented as \cite{al2014optimal}:
\begin{equation}
    \label{eq:PLOS}
	{P_{LOS}(h,r_j) = \frac{1}{1 + a\exp\left(-b\left(\arctan\left(\frac{h}{r_j}\right)-a\right)\right)}},
\end{equation}
where $a \ \mathrm{and} \ b$ are constant values based on the choice of the urban environment.  Finally, the corresponding aerial channel gain is given by
\begin{align}
    \label{eq:aerial_gain}
    G^a(h,r_i)=10^{-0.1 L^a_{dB}(h,r_i)}.
\end{align}

\section{Terrestrial and Aerial NOMA vs. OMA}
In this section, we first briefly review the uplink and downlink NOMA principles for uplink and downlink scenarios. Then, we characterize the performance of users in NOMA and OMA considering terrestrial (ground BS) and aerial (UAV BS) networks for both the uplink and downlink scenarios. We consider Rician fading for aerial NOMA network and Rayleigh fading for terrestrial NOMA network.

\subsection{Principles of Uplink and Downlink NOMA}
Consider a network consisting of a cluster of NOMA users. 
In the downlink scenario,
the total downlink transmit power of BS for a cluster of NOMA users is limited to a maximum allowed power level. The BS transmits the super-imposed signal of NOMA  users, while allocating higher and lower power levels to far and near  users, respectively. Here,
the near  user performs successive interference cancellation (SIC). Subsequently, the highest channel gain user cancels all intra-cluster interferences, whereas the lowest channel gain user receives the interferences from all users within its cluster.

Different from the downlink scenario where the total downlink transmit power is limited, in the uplink scenario, the transmit power level of each user is independently limited to the user's maximum allowed power. Thus, the receiving power from the strongest user is likely the strongest at the BS. Therefore, the strongest user is decoded first at the BS and thus this user experiences the interference from all relatively weaker users in the same NOMA cluster, whereas the weakest user receives no interference from other users.

\subsection{Terrestrial and Aerial OMA: Spectral Efficiency}
 For both the uplink and downlink OMA transmissions, the channel capacity of user 
$i$ for the terrestrial and aerial scenarios is obtained, respectively, as follows:
\begin{align}
    C_i^{(\mathrm{oma,t})} &= 0.5\log_2(1+ P_i G^t(d_i)\Phi_i), \ \forall  i=1,2 
    \\
     C_i^{(\mathrm{oma,a})} &= 0.5\log_2(1+P_i G^a(h,r_i)\chi_i), \ \forall  i=1,2 
\end{align}
where $P_i$ is the normalized uplink (or downlink) transmit power of for user $i$ with respect to the receiver noise power  (i.e., $P_i=P_i^t/n_0$ where $P_i^t$ is the uplink (or downlink) transmit power for user $i$ and $n_0$ is the receiver noise power), $G^t$ and $G^a$ denote  the terrestrial and aerial channel power gains given by \eqref{eq:terrestrial_gain} and \eqref{eq:aerial_gain}, respectively, and $\Phi_i$ and $\chi_i$ are  Rayleigh and Rician fading for user $i$ respectively. Here, the coefficient $0.5$ is used to denote that half of the time is allocated to each user. For example if the time-division-multiple-access (TDMA)  is employed, each of the two users is only allowed to access half of the time resource.  Otherwise the orthogonality assumption is not valid.

\subsection{Terrestrial and Aerial NOMA Spectral Efficiency}

In what follows, for the terrestrial and aerial cases, we study the downlink and uplink NOMA separately.
\subsubsection{Downlink NOMA} Let $U_1$ and $U_2$ be strong and weak users respectively (i.e., users with strong and weak channels are referred as near and far users respectively). Also let $P=a_1P+a_2P$ be normalized aggregate downlink BS transmit power for $U_1$ and $U_2$ with respect to noise power, where $a_1P$ and $a_2P$ are normalized transmit powers for $U_1$ and $U_2$, respectively. Note that we have $a_1+a_2=1$.

The terrestrial downlink NOMA spectral efficiencies of $U_1$ and $U_2$ are given by
\begin{align}
    \label{eq:user1_terrestria_noma}
    \underline{C}_1^{(\mathrm{noma,t})} &= \log_2(1+a_1 P G^t(d_1)\Phi_1), 
    \\
    \label{eq:user2_terrestria_noma}
    \underline{C}_2^{(\mathrm{noma,t})} &= \log_2\left(1+\frac{a_2 P G^t(d_2)\Phi_2}{a_1 P G^t( d_2)\Phi_2 + 1}\right),
\end{align}
and the corresponding spectral efficiencies of $U_1$ and $U_2$ in aerial downlink NOMA are given as
\begin{align}
    \label{eq:user1_aerial_noma}
    \underline{C}_1^{(\mathrm{noma,a})} &= \log_2(1+a_1 P G^a(h,r_1)\chi_1), 
    \\
    \label{eq:user2_aerial_noma}
    \underline{C}_2^{(\mathrm{noma,a})} &= \log_2\left(1+\frac{a_2 P G^a(h,r_2)\chi_2}{a_1 P G^a(h, r_2)\chi_2 + 1}\right).
\end{align}

As seen in \eqref{eq:user1_terrestria_noma} and \eqref{eq:user1_aerial_noma}, $U_1$ is exposed to no interference since it performs SIC. On the other hand, the weak user (i.e., $U_2$) is exposed to the interference from the downlink transmission of $U_1$, as can be seen from \eqref{eq:user2_terrestria_noma} and \eqref{eq:user2_aerial_noma}. The users' spectral efficiencies in aerial and terrestrial NOMA are distinct due to their path-loss models and fading channels.

\subsubsection{Uplink NOMA}
In uplink NOMA, there exists two major differences compared to the downlink NOMA. That is,  each user has his own transmit power constraint, thus we do not have the power allocation coefficients as in the downlink NOMA. Besides, as opposed to the downlink NOMA, here the stronger user receives interference from the weaker user as the BS decodes the signal of strongest user first. Therefore, the terrestrial spectral efficiencies for uplink NOMA are expressed in the following equations:
\begin{align}
    \overline{C}_1^{(\mathrm{noma,t})} &= \log_2\left(1+\frac{P_1G^t(d_1)\Phi_1}{P_2 G^t(d_2)\Phi_2+1}\right), 
    \\
    \overline{C}_2^{(\mathrm{noma,t})} &= \log_2(1+P_2 G^t(d_2)\Phi_2),
\end{align}
and the corresponding spectral efficiencies for aerial uplink NOMA are given by
\begin{align}
    \label{eq:user1_aerial_noma_up}
    \overline{C}_1^{(\mathrm{noma,a})} &= \log_2\left(1+\frac{P_1G^a(h,r_1)\chi_1}{P_2 G^a(h,r_2)\chi_2+1}\right), \\
    \label{eq:user2_aerial_noma_up}
    \overline{C}_2^{(\mathrm{noma,a})} &= \log_2(1+P_2 G^a(h,r_2)\chi_2).
\end{align}

\subsection{Analytical Outage Probability Expressions}
Let $R^{th}$ be the target spectral efficiency threshold of user $i\in\{1,2\}$. The NOMA outage probability is then defined as the probability that user $i$ does not achieve its target-spectral efficiency threshold, i.e.,
\begin{align}
     P_{out,i}^{\mathrm{(noma)}} = \mathrm{\textbf{Pr}}[C_i < R^{th}],
\end{align}
where $C_i$ can be either of the spectral efficiencies of uplink/downlink terrestrial/aerial NOMA expressed in the previous section. In what follows, we formally express the analytical terrestrial and aerial outage probabilities of downlink and uplink NOMA separately.

\subsubsection{Downlink NOMA - Terrestrial and Aerial}
By using the downlink terrestrial spectral efficiencies defined in \eqref{eq:user1_terrestria_noma} and \eqref{eq:user2_terrestria_noma}, and aerial  spectral efficiencies defined in \eqref{eq:user1_aerial_noma} and \eqref{eq:user2_aerial_noma}, the outage probability of users in aerial NOMA and terrestrial NOMA can be derived as in the following lemma.
\begin{lemma}
     For a given channel realization of user $i$ with the terrestrial path-loss $|G^t(d_i)|$ and Rayleigh fading, the downlink outage probability for user $i\in\{1,2\}$ in terrestrial NOMA can be obtained as follows:
    \begin{align}                 
    \label{eq:pout_noma_downlink_terrestrial}
    \underline{P}_{out,i}^{\mathrm{(noma,t)}}=F_{\Phi_i}(\beta_i), \ \ i\in\{1,2\}
         \end{align}
     where
     $
        \beta_1=\dfrac{2^{R^{th}}-1}{a_1 P|G^t(d_1)|}
        ,
        \beta_2=\dfrac{(2^{R^{th}}-1)/P|G^t(d_2)|}{ (a_2 -(2^{R^{th}}-1)a_1 )}
     $, and $F_{\Phi_i}$ is given by \eqref{eq:exponential_cdf}. Similarly,
     for a given  large-scale channel realization $|G^a(h,r_i)|$, the aerial downlink outage probability for user $i\in\{1,2\}$ with Rician fading can be obtained as:
     \begin{align}  
     \label{eq:pout_noma_downlink_aerial}
     \underline{P}_{out,i}^{\mathrm{(noma,a)}}=F_{\chi_i}(\beta_i), \ \ i\in\{1,2\}
     \end{align}
     where
     $ \beta_1=\dfrac{2^{R^{th}}-1}{a_1 P|G^a(h,r_1)|}
        ,
        \beta_2=\dfrac{(2^{R^{th}}-1)/P|G^a(h,r_2)|}{ (a_2 -(2^{R^{th}}-1)a_1 )}
     $, and $F_{\chi_i}$ is given by \eqref{eq:rician_cdf}.
\end{lemma}

\subsubsection{Uplink NOMA for Near User - Terrestrial and Aerial} Here, we derive analytical expressions for the outage probability of the near user in the uplink terrestrial and aerial NOMA.
\begin{theorem}
    \label{th:pout_noma_downlink_1}
    For a given channel realization of near user $U_1$, i.e., $|G^t(d_i)|$, the uplink terrestrial NOMA outage probability  denoted by $\overline{P}_{out,1}^{\mathrm{(noma,t)}}$ is analytically obtained by the following:
    \begin{align}
       \label{eq:pout_noma_uplink_terrestrial_user1} \overline{P}_{out,1}^{\mathrm{(noma,t)}}
         = 1- \frac{\lambda_2 e^{-\alpha_2 \lambda_1}}{\lambda_2+\alpha_1\lambda_1},
    \end{align}
    in which $\alpha_1=(2^{R^{th}} -1)P_2|G^t(d_2)|/P_1|G^t(d_1)|$ and $\alpha_2=(2^{R^{th}} -1)/P_1|G^t(d_1)|$.
\end{theorem}
\begin{proof}
    For a given channel realization of $U_1$ we have:
     \begin{align*}
        \overline{P}_{out,1}^{\mathrm{(noma,t)}}
        &=
        \textrm{Pr}\{\overline{C}_1^{\mathrm{(noma,t)}}<R^{th}\}
        \\
        &=
        \textrm{Pr}\left\{ \frac{P_1 |G^t(d_1)|\Phi_1}{P_2 |G^t(d_2)|\Phi_2+1}<2^{R^{th}}-1
        \right\}
        \\
        &=
        \textrm{Pr}\left\{ y<\alpha_1 x + \alpha_2 \right\},
     \end{align*}
     where $y=\Phi_1$ and $x=\Phi_2$ are exponentially distributed variables associated with $U_1$ and $U_2$, respectively, 
     and $\alpha_1=(2^{R^{th}} -1)P_2|G^t(d_2)|/P_1|G^a(d_1)|$ and $\alpha_2=(2^{R^{th}} -1)/P_1|G^t(d_1)|$. By considering that $x$ and $y$ are uncorrelated random variables, we have 
      \begin{align*}
        \overline{P}_{out,1}^{\mathrm{(noma,t)}}
			&= \textrm{Pr}\{ y < \alpha_1 x + \alpha_2  \} 
			\\&=
			\int_{0}^{\infty} \int_{0}^{\alpha_1 x + \alpha_2} f_{\phi_1}(y) f_{\phi_2}(x) dy dx
			\\&=
			\lambda_1 \lambda_2
			\int_{0}^{\infty} \int_{0}^{\alpha_1 x + \alpha_2} e^{-\lambda_1y - \lambda_2x} dy dx
		    \\&=
		    1- \frac{\lambda_2 e^{-\alpha_2 \lambda_1}}{\lambda_2+\alpha_1\lambda_1}. 
	\end{align*}
\end{proof}
\begin{theorem}
    \label{th:pout_noma_uplink_1}
    For a given channel realization  $|G^a(h,r_i)|$, the uplink aerial NOMA outage probability for $U_1$ denoted by $\overline{P}_{out,1}^{\mathrm{(noma,a)}}$ is analytically obtained by the following equation:
    \begin{align}
       \label{eq:pouttttttt} &\overline{P}_{out,1}^{\mathrm{(noma,a)}}
         =
       n_1 n_2 e^{-(K_1+K_2)}
            \sum_{k_1=0}^{\infty} 
            \sum_{k_2=0}^{\infty}
            \frac{m_1^{k_1} m_2^{k_2}  }
            {4^{k_1+k_2} {k_1}!{k_2}! } \times \nonumber
                \\
                &
                \hspace{10pt}
            \Bigg[ 
                \frac{1}{n_1^{k_1+1} n_2^{k_2+1}}
                  - \frac{\alpha_2^{k_1+k_2+1}}{e^{n_2 \alpha_2}\alpha_1^{k_1+1}} \sum_{k_3=0}^{k_2} \frac{\alpha_2^{-k_3} n_2^{-k_3-1} }{(k_2-k_3)!} \times \nonumber
            \\ 
                &                 \hspace{25pt} \Psi(k_1+1,k_1 +k_2-k_3+2, (n_1+n_2\alpha_1)(\frac{\alpha_2}{\alpha_1})  
            \Bigg],
    \end{align}
    in which $n_i=\frac{1+K_i}{\Omega_i}, m_i=\frac{4 K_i(1+K_i)}{\Omega_i}, i=1,2$, and  $\Psi$ is Tricomi Confluent Hyper-Geometric function and $\alpha_1=(2^{R^{th}} -1)P_2|G^a(h,r_2)|/P_1|G^a(h,r_1)|$ and $\alpha_2=(2^{R^{th}} -1)/P_1|G^a(h,r_1)|$.
    
    \begin{proof}
        See \textbf{Appendix \ref{apx:proof_of_th_pout_noma_uplink_1}}.
    \end{proof}
    
\end{theorem}
The following lemma obtains the uplink terrestrial and aerial outage probability of $U_2$ which is the farther user.
\begin{lemma}
     For a given  large-scale channel realization  $|G^t(d_2)|$, the terrestrial uplink outage probability for $U_2$ is obtained as follows:
     \begin{align}
     \label{poutmmmmm}
     \underline{P}_{out,2}^{\mathrm{(noma,t)}}=F_{\Phi_2}(\alpha),
     \end{align}
     where $\alpha={(2^{R^{th}} -1)}/{P_2|G^t(d_2)|}$ and $F_{\Phi_2}$ is given by \eqref{eq:exponential_cdf}. Besides, for a given  large-scale channel realization  $|G^a(h,r_2)|$, the aerial uplink outage probability for $U_2$ is obtained as
     \begin{align}         
     \label{eq:pout_downlink_aerial_user2}
     \underline{P}_{out,2}^{\mathrm{(noma,a)}}=F_{\chi_2}(\alpha),
     \end{align}
     where $\alpha={(2^{R^{th}} -1)}/{P_2|G^a(h,r_2)|}$ and $F_{\chi_2}$ is given by \eqref{eq:rician_cdf}.
\end{lemma}

\section{Numerical Results and Discussions}

In this section, we compare the performance of aerial and terrestrial NOMA and OMA for a variety of network parameters. Consider a circular area of radius $500$ m. The noise power spectral density is assumed to be $10^{-10}$ W/Hz, carrier frequency is $2.5$ GHz, $\eta_{LOS}=1.6$, $\eta_{NLOS}=23$, and the coefficients $a$ and $b$ in \eqref{eq:PLOS} are $12.8$ and $0.11$, respectively.
In all simulation scenarios, we consider  
the Rayleigh fading parameter $\lambda_i$ and the Rician fading power $\Omega_i$ are equal to unity for $i=\{1,2\}$.  For each snapshot of the simulation, $100$ users are considered to be randomly scattered in the cell area.

The distance of each user $i$ to the BS (i.e., $r_i$) is a uniformly distributed random variable (i.e., $r_i \sim U(0,R)$ in which $R$ is the cell radius). By using order statistics, the PDF and CDF of the near and far users are obtained, respectively, as follows:
\begin{subequations}
\begin{align}
    f(r_{\mathrm{min}})&=f(r_1)=N \left[1-F(r) \right]^{N-1} f(r)
    \\
    f(r_{\mathrm{max}})&=f(r_2)=N F(r)^{N-1} f(r)
\end{align}
\end{subequations}
where $f(r) = 2 r/R^2$ and $F(r)=r^2/R^2$ are the PDF and CDF of uniformly distributed variable $r\sim U(0,R)$. Now, by averaging the conditional outage probability $P_{out,1}^{\mathrm{noma}}(r_{1} ,r_{2})$ of near user in uplink NOMA  \eqref{eq:pout_noma_uplink_terrestrial_user1} and \eqref{eq:pouttttttt}
over the joint PDF of $r_{1}$ and $r_{2}$, we compute the outage numerically through standard mathematical software  \texttt{Mathematica}.
\begin{align}
    \label{eq:pout_noma_versus_r}
   & \mathbb{E}_{r_1,r_2}\{P_{out,1}^{\mathrm{(noma)}}(r_{1},r_{2})\}  
   \nonumber\\&=\int \int f(r_{1},r_{2}) P_{out,1}^{\mathrm{(noma)}}(r_{1},r_{2}) dr_{1}dr_{2}, \quad i\in\{1,2\}
   \nonumber\\&\approx \int_0^R \int_0^R   f(r_{1}) f(r_{2}) P_{out,1}^{\mathrm{(noma)}}(r_{1},r_{2}) dr_{1}dr_{2}.
\end{align}
Note that the ranked variables $r_1$ and $r_2$ are not independent; however, there dependence is relatively weak \cite{tabassum2015spectral}. Therefore, we approximate the joint PDF $f(r_{1},r_{2}) \approx f(r_{1}) f(r_{2})$. For the far user in uplink NOMA, we compute the outage probability as $\mathbb{E}_{r_2}\{P_{out,2}^{\mathrm{(noma)}}(r_{2})\} $. On the other hand, for downlink NOMA, the outage probabilities of near and far users can be given as  $\mathbb{E}_{r_1}\{P_{out,1}^{\mathrm{(noma)}}(r_{1})\} $ and  $\mathbb{E}_{r_2}\{P_{out,2}^{\mathrm{(noma)}}(r_{2})\} $, respectively. 

\subsection{Validation of derived
\label{sec:sim_closed_form}
closed-form outage probabilities}
\label{sec_sim1}
\begin{figure}
    \centering
    \includegraphics[width=\linewidth]{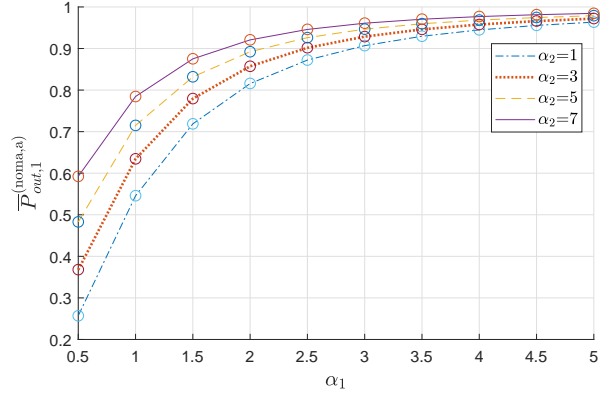}
    \caption{Comparison of the outage probabilities versus $\alpha_1$ and $\alpha_2$ through derived expression in  \eqref{eq:pouttttttt} and  corresponding values obtained through simulation. Analytical values are shown by lines and  Monte-carlo simulations are shown in circles.}
    \label{fig:outage_verify}
\end{figure}
{Since the expressions in
\eqref{eq:pout_noma_downlink_terrestrial}-\eqref{eq:pout_downlink_aerial_user2} are special cases of \eqref{eq:pouttttttt}}, we validate \eqref{eq:pouttttttt} through Monte-Carlo simulations. Fig. \ref{fig:outage_verify} depicts the outage probability  $\overline{P}^{(\mathrm{noma,a})}_{out,1}$ in \eqref{eq:pouttttttt} 
for different values of $\alpha_1$ and $\alpha_2$. It is seen that the values obtained through derived expressions (shown by lines), exactly match those obtained through Monte-Carlo simulations (shown in circles).

\subsection{Downlink NOMA: Terrestrial vs Aerial}
\label{sec:sim_downlink}

\begin{figure*}[htp]
  \subfigure[aerial scenario]{ \includegraphics[width=0.5\textwidth]{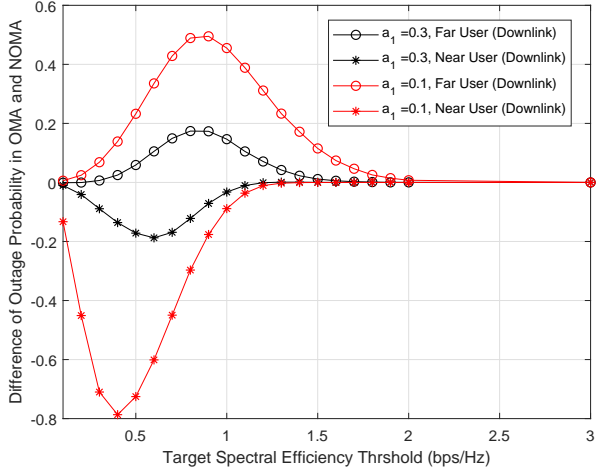}
    }
    \subfigure[terrestrial scenario]{
    \includegraphics[width=0.5\textwidth]{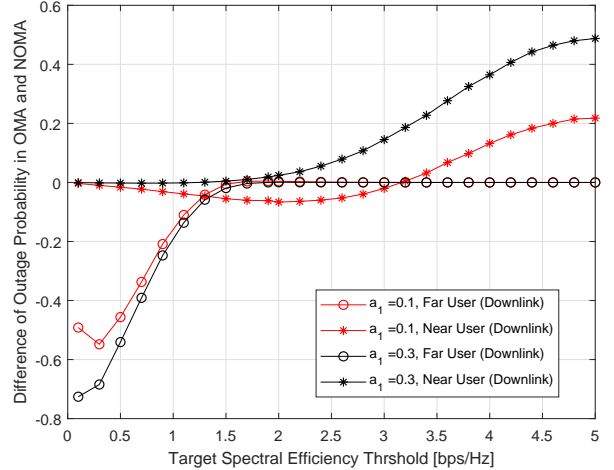}
        }
  \caption{Comparison of downlink aerial and terrestrial NOMA and OMA outage probabilities for near and far users versus varying target spectral efficiency threshold and power allocation coefficient $a_1$, UAV altitude $h=1500$~m, BS total transmit power $P=5$~W, Rice-K factor  $K_i$=10 and Rayleigh fading factor $\lambda_i=1$ for $i=\{1,2\}$.}
  \label{fig:downlink_aerial_terrestrial_vs_target_rate}
\end{figure*}

\begin{figure}
    \centering
    \includegraphics[scale=0.5]{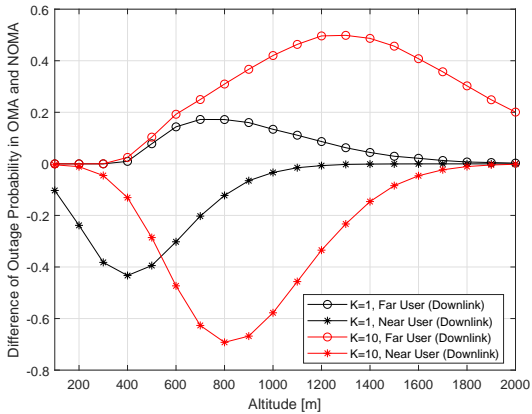}
    \caption{Comparison of downlink aerial NOMA and OMA outage probabilities for near and far users versus the altitude of UAV and different values of Rician $K$ factor, UAV altitude $h=1500$~m, BS total transmit power $P=5$ W, and target spectral efficiency threshold  $R^{th}=1$ bps/Hz.}
    \label{fig:downlink_aerial_versus_altitude}
\end{figure}
In order to compare the performance of NOMA and OMA for user $i$,  we define the performance gain of NOMA over OMA in aerial network $\eta^a$ as follows:
\begin{equation}
    \eta_i^a = \overline{P}^{(\mathrm{oma,a})}_{out,i}- \overline{P}^{(\mathrm{noma,a})}_{out,i}
\end{equation}
Evidently, if $ \eta_i^a > 0$ it demonstrates the superior performance of NOMA over OMA, whereas if $ \eta_i^a < 0$ OMA outperforms NOMA.
Fig. \ref{fig:downlink_aerial_terrestrial_vs_target_rate} compares the performance of  downlink NOMA for near and far users versus target spectral efficiency threshold (i.e., $R^{th}$) for different values of power allocation  coefficient of near user $a_1$. Surprisingly, we note  that for aerial scenario, the far user always enjoys positive NOMA gain  while the near user experiences negative NOMA gain. For example for $R^{th}=1$ and $a_1=0.1$, the NOMA gain of far user is $0.47$, while that of the near user is only $-0.09$. For terrestrial scenario,  no user experiences positive NOMA gain for rather low values of $R^{th}$, and only near user can obtain reasonable NOMA gain for relatively high values of $R^{th}$ due to  SIC. That is, due to the high interference in terrestrial scenario, far user can not generally benefit from NOMA, and the near user can only benefit from NOMA for rather high values of $R^{th}$. As opposed to terrestrial network, we note  that the far user in aerial downlink can  benefit from NOMA for lower values of $R^{th}$ compared to the near user. The reason is low interference due to farther transmission distances and LOS transmissions. For some values of $R^{th}$, the positive NOMA gain of far user notably dominates the negative NOMA gain of near user.

Fig. \ref{fig:downlink_aerial_versus_altitude} shows how the altitude of UAV and Rician $K$ factor affect the performance of NOMA in aerial downlink scenario. We note that there always exists some optimal UAV height (i.e., $h$) for which users yield the best NOMA gain. For example, while it is seen that for $K=10$, near user shows the best NOMA gain with $h=1200$ m, the overall performance seems to be better at $h=1500$ m, wherein near user and far user experience the NOMA gain of $+0.45$ and $-0.08$, respectively. Besides, it is seen that far user can benefit more from NOMA with higher values of $K$, i.e., stronger LOS.
   
\subsection{Uplink NOMA: Terrestrial vs Aerial} 
\label{sec:sim_uplink}

\begin{figure*}[htp]
  \subfigure[aerial scenario]{ \includegraphics[width=0.5\textwidth]{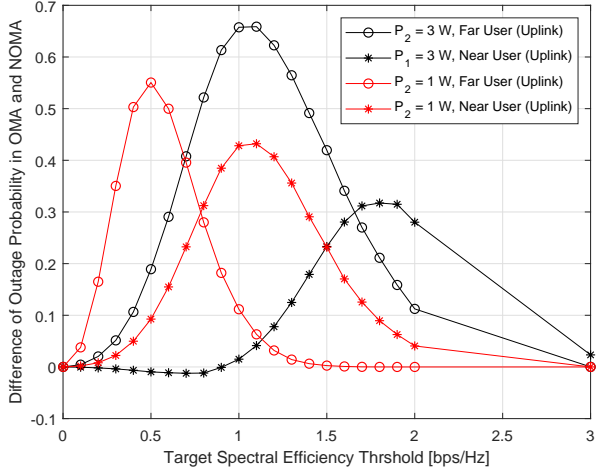}
    }
    \subfigure[terrestrial scenario]{
    \includegraphics[width=0.5\textwidth]{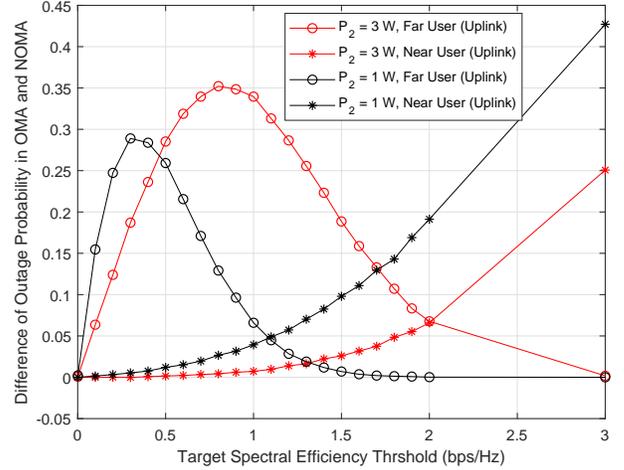}
        }
  \caption{Comparison of uplink aerial and terrestrial NOMA and OMA outage probabilities for near and far users versus varying target spectral efficiency threshold and power allocation for $P_1$ and $P_2$, UAV altitude $h=600$~m, and Rice-K factor $K_i=10$ and Rayleigh fading factor $\lambda_i=1$ for $i=\{1,2\}$.}
  \label{fig:uplink_aerial_terrestrial_vs_target_rate}
\end{figure*}

\begin{figure}
    \centering
    \includegraphics[scale=0.5]{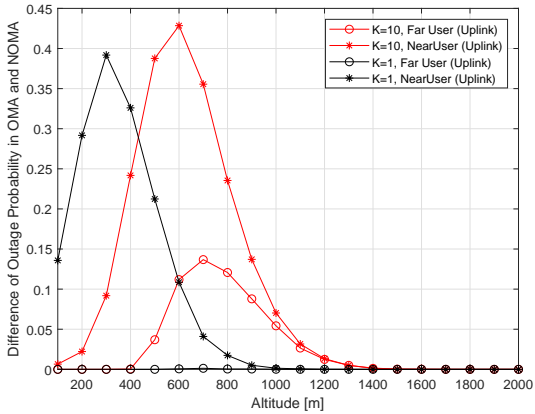}
    \caption{Difference of OMA and NOMA uplink aerial  outage probabilities for near and far users versus the altitude of UAV and different values of Rician $K$ factor, UAV altitude $h=600$~m,  transmit power per user $P_i=1$~W, and target spectral efficiency threshold  $R^{th}=1$ bps/Hz.}
    \label{fig:uplink_aerial_versus_altitude}
\end{figure}
Similar to Fig. \ref{fig:downlink_aerial_terrestrial_vs_target_rate} employed for the downlink NOMA,
Fig. \ref{fig:uplink_aerial_terrestrial_vs_target_rate} compares the performance of uplink NOMA and OMA versus the target spectral efficiency threshold $R^{th}$ and different values of transmit powers for near and far users. Firstly, it is seen that unlike the downlink NOMA wherein the positive NOMA gain of one user mostly leads to the negative NOMA gain of the user, for the uplink scenario, both users can generally benefit simultaneously from NOMA for both terrestrial and aerial cases. 
It is also seen that increasing the transmit power of both near and far users in either of the aerial and terrestrial scenarios, results  in the degradation  of the performance of far and near user. Besides, since the far user in terrestrial network experiences very poor channel, as shown in the figure, such user can only benefit from NOMA for rather high values of the target spectral efficiency threshold.


Fig. \eqref{fig:uplink_aerial_versus_altitude} shows the NOMA gain of uplink aerial scenario for different values of Rician $K$ factor versus the altitude of UAV. Similar to the downlink scenario, it is seen that increasing the value of Rician $K$ factor improves the performance of NOMA (by increasing the LOS strong component of signal). More specifically, for $K=1$, far user receives negligible NOMA gain, while for $K=10$, both users benefit  from NOMA compared to the case for $K=1$. It is also seen that the optimal height of UAV depends also on the value of $K$, i.e., for higher values of $K$, the LOS component of signal has more impact on the performance, and since the LOS probability increases with the height of UAV, it is seen that for higher values of $K$, the NOMA optimal performance is experienced on higher values of UAV altitude. 

\section{Conclusion}
In this work, we provided theoretical expressions and numerical results to compare the outage performance of NOMA and OMA in aerial vs. terrestrial networks considering both the uplink and downlink scenarios. The comparison of NOMA gain for aerial networks  and terrestrial networks was explored considering different parameters such as fading shape factor, transmit power level, target spectral efficiency threshold and UAV altitude. For example, it was  shown through numerical results that the aerial network yields more NOMA gain  for low values of the target spectral efficiency threshold, while the trend is reverse for  high values of target spectral efficiency threshold. 


\begin{appendices}
    \section{Proof of Theorem \ref{th:pout_noma_uplink_1}}
    \label{apx:proof_of_th_pout_noma_uplink_1}
     The outage probability  of $U_i$ is given as
     \begin{align}
     \label{eq:t23}
        \overline{P}_{out,i}^{\mathrm{(noma,a)}}
        &=
        \textrm{Pr}\{\overline{C}_1^{\mathrm{(noma,a)}}<R^{th}\}
        \\
        &=
        \textrm{Pr}\left\{ \frac{P_1 |G^a(h,r_1)|\chi_1}{P_2 |G^a(h,r_2)|\chi_2+1}<2^{R^{th}}-1
        \right\}
        \\
        &=
        \textrm{Pr}\left\{ y<\alpha_1 x + \alpha_2 \right\}, 
     \end{align}
     where $y=\chi_1$ and $x=\chi_2$ are the Rician fading variables associated with $U_1$ and $U_2$, respectively 
     and $\alpha_1=(2^{R^{th}} -1)P_2|G^a(h,r_2)|/P_1|G^a(h,r_1)|$ and $\alpha_2=(2^{R^{th}} -1)/P_1|G^a(h,r_1)|$.
     By considering $n_i=\frac{1+K_i}{\Omega_i}$, $m_i=\frac{4 K_i(1+K_i)}{\Omega_i}$, $C=n_1n_2e^{-(K_1+K_2)}$, and using the Taylor series of the Modified Bessel function of $I_0(x)=\sum_{k=0}^{\infty} \frac{(x/2)^{2k}}{(k!)^2}$, and by considering that $y$ and $x$ are Rician distributed random variables corresponding to $U_1$ and $U_2$, respectively,  \eqref{eq:t23} can be written as in the following.

    \begin{align*}
        \label{eq:pout_noma1_proof}
        &\overline{P}_{out,1}^{\mathrm{(noma,a)}}
			= \textrm{Pr}\{ y < \alpha_1 x + \alpha_2  \} 
			\\
			&=  \int_{0}^{\infty} \int_{0}^{\alpha_1 x + \alpha_2} f_{\chi_1}(y) f_{\chi_2}(x) dy dx
			\\
			&=	
			 C
			 \sum_{k_1=0}^{\infty} 
			 \sum_{k_2=0}^{\infty}
			 \int_{0}^{\infty}
			 \frac{e^{-n_1 x} m_1^{k_1} m_2^{k_2} x^{k_1}}
			{4^{k_1+k_2} ({k_1}!{k_2}!)^2}
			\int_{0}^{\alpha_1 x + \alpha_2} 
	        \hspace{-10pt} e^{-n_2 y}
			y^{k_2} dy dx
	    \\&=	
            C
            \sum_{k_1=0}^{\infty} 
            \sum_{k_2=0}^{\infty}
            \int_{0}^{\infty}
             \frac{e^{- n_1 x} m_1^{k_1} m_2^{k_2} x^{k_1}}
            {4^{k_1+k_2} ({k_1}!{k_2}!)^2}
            \times
            \\
            & \hspace{80pt}
            \left[
            e^{-n_2 y} \sum_{k_3=0}^{k_2} \frac{ - k_3!\binom{k_2}{k_3}}{(n_2)^{k_3+1}} y^{k_2-k_3}
            \right]_0^{\alpha_1 x + \alpha_2}
            dx
         \\&=
            C
            \sum_{k_1=0}^{\infty} 
            \sum_{k_2=0}^{\infty}
            \frac{m_1^{k_1} m_2^{k_2}  }
            {4^{k_1+k_2} ({k_1}!{k_2}!)^2 } 
            \Bigg[
                \frac{k_2!}{n_2^{k_2+1}}
                \int_{0}^{\infty}
                e^{- n_1  x  }  x^{k_1}
                dx  -
                \\
                &
                 \sum_{k_3=0}^{k_2} \frac{k_3!\binom{k_2}{k_3}}{n_2^{k_3+1}}e^{-n_2 \alpha_2} \!
                \int_{0}^{\infty} \!\!
                e^{- (n_1 + n_2\alpha_1)  x  }  x^{k_1}
                (\alpha_1 x +\alpha_2)^{k_2-k_3}
                dx \Bigg]
         \\&=
            C
            \sum_{k_1=0}^{\infty} 
            \sum_{k_2=0}^{\infty}
            \frac{m_1^{k_1} m_2^{k_2}  }
            {4^{k_1+k_2} {k_1}!{k_2}! } \times
                \\
                & \hspace{15pt}
                \Bigg[ 
                \frac{1}{n_1^{k_1+1} n_2^{k_2+1}}
                  - \frac{\alpha_2^{k_1+k_2+1}}{e^{n_2 \alpha_2}\alpha_1^{k_1+1}} \sum_{k_3=0}^{k_2} \frac{\alpha_2^{-k_3} n_2^{-k_3-1} }{(k_2-k_3)!} \times
                  \\
                  &\hspace{42pt}
                \Psi\left(k_1+1,k_1 +k_2-k_3+2, (n_1+n_2\alpha_1)(\frac{\alpha_2}{\alpha_1}) \right) 
                \Bigg]. 
    \end{align*}
    This completes the proof.

     \bibliographystyle{IEEEtran}
\bibliography{Uplink}

\begin{thebibliography}{10}
\providecommand{\url}[1]{#1}
\csname url@rmstyle\endcsname
\providecommand{\newblock}{\relax}
\providecommand{\bibinfo}[2]{#2}
\providecommand\BIBentrySTDinterwordspacing{\spaceskip=0pt\relax}
\providecommand\BIBentryALTinterwordstretchfactor{4}
\providecommand\BIBentryALTinterwordspacing{\spaceskip=\fontdimen2\font plus
\BIBentryALTinterwordstretchfactor\fontdimen3\font minus
  \fontdimen4\font\relax}
\providecommand\BIBforeignlanguage[2]{{%
\expandafter\ifx\csname l@#1\endcsname\relax
\typeout{** WARNING: IEEEtran.bst: No hyphenation pattern has been}%
\typeout{** loaded for the language `#1'. Using the pattern for}%
\typeout{** the default language instead.}%
\else
\language=\csname l@#1\endcsname
\fi
#2}}

\bibitem{li2018uav}
B.~Li, Z.~Fei, and Y.~Zhang, ``{UAV communications for 5G and beyond: Recent
  advances and future trends},'' \emph{IEEE Internet of Things Journal},
  vol.~6, no.~2, pp. 2241--2263, 2018.

\bibitem{hossain2014evolution}
E.~Hossain, M.~Rasti, H.~Tabassum, and A.~Abdelnasser, ``{Evolution toward 5G
  multi-tier cellular wireless networks: An interference management
  perspective},'' \emph{IEEE Wireless Communications}, vol.~21, no.~3, pp.
  118--127, 2014.

\bibitem{sharma2017uav}
P.~K. Sharma and D.~I. Kim, ``{UAV-enabled downlink wireless system with
  non-orthogonal multiple access},'' in \emph{2017 IEEE Globecom Workshops (GC
  Wkshps)}.\hskip 1em plus 0.5em minus 0.4em\relax IEEE, 2017, pp. 1--6.

\bibitem{sekander2018multi}
S.~Sekander, H.~Tabassum, and E.~Hossain, ``{Multi-Tier drone architecture for
  5G/B5G cellular networks: challenges, trends, and prospects},'' \emph{IEEE
  Communications Magazine}, vol.~56, no.~3, pp. 96--103, 2018.

\bibitem{originalNOMA}
Z.~{Ding}, Y.~{Liu}, J.~{Choi}, Q.~{Sun}, M.~{Elkashlan}, C.~{I}, and H.~V.
  {Poor}, ``{Application of Non-Orthogonal Multiple Access in LTE and 5G
  networks},'' \emph{IEEE Communications Magazine}, vol.~55, no.~2, pp.
  185--191, February 2017.

\bibitem{tabassum2017uplink}
H.~Tabassum, M.~S. Ali, E.~Hossain, M.~J. Hossain, and D.~I. Kim, ``{Uplink vs.
  downlink NOMA in cellular networks: Challenges and research directions},'' in
  \emph{2017 IEEE 85th Vehicular Technology Conference (VTC Spring)}.\hskip 1em
  plus 0.5em minus 0.4em\relax IEEE, 2017, pp. 1--7.

\bibitem{rupasinghe2017non}
N.~Rupasinghe, Y.~Yapici, I.~G{\"u}ven{\c{c}}, and Y.~Kakishima,
  ``{Non-orthogonal multiple access for mmWave drones with multi-antenna
  transmission},'' in \emph{2017 51st Asilomar Conference on Signals, Systems,
  and Computers}.\hskip 1em plus 0.5em minus 0.4em\relax IEEE, 2017, pp.
  958--963.

\bibitem{nasir2019uav}
A.~A. Nasir, H.~D. Tuan, T.~Q. Duong, and H.~V. Poor, ``{UAV-enabled
  communication using NOMA},'' \emph{IEEE Transactions on Communications},
  2019.

\bibitem{sohail2018non}
M.~F. Sohail, C.~Y. Leow, and S.~Won, ``{Non-orthogonal multiple access for
  unmanned aerial vehicle assisted communication},'' \emph{IEEE Access},
  vol.~6, pp. 22\,716--22\,727, 2018.

\bibitem{cui2019multiple}
F.~Cui, Y.~Cai, Z.~Qin, M.~Zhao, and G.~Y. Li, ``{Multiple access for
  mobile-UAV enabled networks: Joint trajectory design and resource
  allocation},'' \emph{IEEE Transactions on Communications}, 2019.

\bibitem{sun2018cyclical}
J.~Sun, Z.~Wang, and Q.~Huang, ``{Cyclical NOMA based UAV-enabled wireless
  network},'' \emph{IEEE Access}, vol.~7, pp. 4248--4259, 2018.

\bibitem{farajzadeh2019uav}
A.~Farajzadeh, O.~Ercetin, and H.~Yanikomeroglu, ``{UAV data collection over
  NOMA backscatter networks: UAV altitude and trajectory optimization},''
  \emph{arXiv preprint arXiv:1902.03061}, 2019.

\bibitem{duan2019resource}
R.~Duan, J.~Wang, C.~Jiang, H.~Yao, Y.~Ren, and Y.~Qian, ``{Resource allocation
  for multi-UAV aided IoT NOMA uplink transmission systems},'' \emph{IEEE
  Internet of Things Journal}, 2019.

\bibitem{UplinkNOMArandomaccess}
J.~B. {Seo}, S.~{Pack}, and H.~{Jin}, ``{Uplink NOMA random access for
  UAV-Assisted communications},'' \emph{IEEE Transactions on Vehicular
  Technology}, pp. 1--1, 2019.

\bibitem{wirelessrelaynetwork}
F.~{Ono}, H.~{Ochiai}, and R.~{Miura}, ``{A wireless relay network based on
  unmanned aircraft system with rate optimization},'' \emph{IEEE Transactions
  on Wireless Communications}, vol.~15, no.~11, pp. 7699--7708, Nov 2016.

\bibitem{al2014optimal}
A.~Al-Hourani, S.~Kandeepan, and S.~Lardner, ``{Optimal LAP altitude for
  maximum coverage},'' \emph{IEEE Wireless Commun. Letters}, vol.~3, no.~6, pp.
  569--572, 2014.

\bibitem{tabassum2015spectral}
H.~Tabassum, E.~Hossain, M.~J. Hossain, and D.~I. Kim, ``On the spectral
  efficiency of multiuser scheduling in rf-powered uplink cellular networks,''
  \emph{IEEE transactions on wireless communications}, vol.~14, no.~7, pp.
  3586--3600, 2015.

\end{thebibliography}

\end{appendices}

\end{document}